\documentclass[letterpaper, 10 pt, conference]{IEEEconf}  
\usepackage[utf8]{inputenc}

\IEEEoverridecommandlockouts    
\usepackage{cite}
\usepackage{amsmath,amssymb,amsfonts}
\usepackage{graphicx}

\usepackage{algorithm}
\usepackage{algorithmic}

\usepackage{booktabs}

\usepackage{textcomp}
\usepackage{amssymb}
\usepackage{amsmath}
\usepackage{graphicx}
\usepackage{xcolor}
\usepackage{makecell}
\usepackage{comment}
\usepackage{todonotes}

\usepackage{cite}
\usepackage{graphicx}
\usepackage{subcaption}

\usepackage{xcolor}
\definecolor{light-blue}{rgb}{0.3,0.5,0.8}
\usepackage[colorlinks=true, linkcolor=cyan, citecolor=cyan, filecolor=magenta, urlcolor=cyan]{hyperref}
\usepackage{cleveref}

\title{\bf Safe and Efficient Coexistence of Autonomous Vehicles with Human-Driven Traffic at Signalized
Intersections}

\author{Filippos N. Tzortzoglou, \IEEEmembership{Student Member, IEEE,} Logan E. Beaver, \IEEEmembership{Member, IEEE} \\ and Andreas A. Malikopoulos,
\IEEEmembership{Senior Member, IEEE}
\thanks{This research was supported in part by NSF under Grants CNS-2401007, CMMI-2348381, IIS-2415478, and in part by MathWorks.} \thanks{Filippos N. Tzortzoglou, and Andreas A. Malikopoulos are with the Department of Civil and Environmental Engineering, Cornell University, Ithaca, NY 14853 USA.   (emails: \tt\small{ft253@cornell.edu; amaliko@cornell.edu)}}\thanks{Logan E. Beaver is with the Department of Mechanical and Aerospace Engineering, Old Dominion University, Norfolk, VA 23529  (email:\tt\small{ lbeaver@odu.edu)}}}

\usepackage{amsthm}

\newtheorem{theorem}{Theorem}

\newtheorem{remark}{Remark}

\begin{document}

  \maketitle

 \thispagestyle{empty}
 \pagestyle{empty}
 
\begin{abstract}

The proliferation of connected and automated vehicles (CAVs) has positioned mixed traffic environments, which encompass both CAVs and human-driven vehicles (HDVs), as critical components of emerging mobility systems. Signalized intersections are paramount for optimizing transportation efficiency and enhancing fuel economy, as they inherently induce stop-and-go traffic dynamics. In this paper, we present an integrated framework that concurrently optimizes signal timing and CAV trajectories at signalized intersections, with the dual objectives of maximizing traffic throughput and minimizing energy consumption for CAVs. We first formulate an optimal control strategy for CAVs that prioritizes trajectory planning to circumvent state constraints, while incorporating the impact of signal timing and HDV behavior. Furthermore, we introduce a traffic signal control methodology that dynamically adjusts signal phases based on vehicular density per lane, while mitigating disruption for CAVs scheduled to traverse the intersection. Acknowledging the system’s inherent dynamism, we also explore event-triggered replanning mechanisms that enable CAVs to iteratively refine their planned trajectories in response to the emergence of more efficient routing options. The efficacy of our proposed framework is evaluated through comprehensive simulations conducted in MATLAB.
\end{abstract}

\section{INTRODUCTION}
Connected and automated vehicles (CAVs) have attracted considerable attention within the transportation sector \cite{GUANETTI201818}. However, as reported in \cite{alessandrini2015automated}, society cannot expect 100$\%$ penetration of CAVs before 2060. For this reason, researchers are actively exploring the challenging problem of managing mixed traffic environments that include both CAVs and human-driven vehicles (HDVs) \cite{li2023survey}. 

A significant challenge related to mixed traffic is the development of control strategies for the operation of CAVs alongside HDVs at intersections. It has been shown in \cite{VARAIYA2013177} that dynamically selecting the best timing plan for traffic signals based on the current traffic flow can alleviate congestion at intersections, leading to shorter travel times and fewer stop-and-go events, which directly influence energy consumption and emissions. On the other hand, controlling the acceleration of CAVs can reduce emissions, enhance passenger comfort, improve traffic stability and throughput, and ensure safety \cite{hellstrom2010design}. However, most existing literature focuses on adjusting either traffic light phases or vehicle speeds while only a few studies consider the simultaneous control of CAVs and traffic lights.  Next, we review the relevant literature.
\subsection{Literature Review}
\subsubsection{Control of traffic signals}
The dynamic control of traffic signals has been widely studied in the literature. One early attempt along these lines is the \textit{Sydney coordinated adaptive traffic system} \cite{sims1980sydney}, which selects from pre-defined traffic phases to maximize throughput. Similarly, the authors in \cite{MIRCHANDANI2001415} introduced a traffic-adaptive approach that partitions the traffic signal control into subproblems, predicting flows and solving separate optimization problems to adjust signals. Later, in \cite{kulkarni2007} a fuzzy controller 
 was proposed that extends or ends a green phase based on approaching vehicles and queues. Then, a seminal framework called \textit{MaxPressure} was proposed in \cite{VARAIYA2013177} that improves network throughput using local data at signalized intersections. Sequel studies were inspired by \cite{VARAIYA2013177} and extend this approach using reinforcement learning \cite{wei2019presslight}. A broader survey is provided in~\cite{eom2020traffic}.

\subsubsection{Control of CAVs}
Control of CAVs at intersections has also received significant attention \cite{hellstrom2010design,wang2023co}. A central approach is the \textit{green light optimal speed advisory} (GLOSA) system\cite{seredynski2013comparison}, which adjusts CAVs' speed according to signal phases. Field experiments have demonstrated GLOSA’s potential \cite{stahlmann2018exploring}, while recent research also incorporated stochastic switching times for the traffic signals \cite{TYPALDOS2023104364}. Recently, in addition to GLOSA systems, an optimal control approach was presented in \cite{meng2020eco} to define CAVs' trajectories that cross intersections without stopping. A comprehensive review of CAVs control at signalized intersections is presented in \cite{wang2023co}.

\subsubsection{Joint control of traffic signals and CAVs}
Recent efforts also address the joint control of both signals and CAVs. One main category in the literature utilizes bi-level optimization, determining signal phases first and then defining CAV trajectories, based on the traffic signals \cite{du2021coupled}. Another category jointly optimizes signal timing and CAVs, under the same optimization problem, however, real-time execution remains challenging with high traffic \cite{le2024distributed}. Reinforcement learning approaches have also been explored to optimize signal timing and CAV control \cite{guo2023cotv}.


\subsubsection{Contributions of this paper} Although recent studies have explored the joint control of CAVs and traffic signals at signalized intersections, to our knowledge, no prior work has approached this problem using optimal control for CAVs with adaptive signal timing while accounting for future HDVs behavior. This is particularly challenging, as CAVs must account in real time for both signal phase changes and HDVs' trajectories. While some research efforts have employed optimal control, they often rely on simplifying assumptions. For instance, \cite{meng2020eco} assumes fixed signal timings and disregards future HDVs' trajectories, switching to a car-following model when conflicts arise. Similarly, \cite{feng2018spatiotemporal} considers only CAVs, neglecting HDVs altogether. In contrast, we propose an optimal control framework that incorporates future HDVs behavior and dynamic signal timing. We further introduce event-triggered replanning conditions so CAVs adapt efficiently to intersection states. Finally we validate our approach through simulations in MATLAB. 

The remainder of the paper is organized as follows. Section \ref{modeling framework} outlines the modeling framework. Section \ref{CAVs control} details the CAV control strategy under traffic light and safety constraints. Section \ref{Human_Driven_vehicles_section} introduces the human driving model, and Section \ref{traffic lights control} describes the traffic light control approach. Simulation results are presented in Section \ref{simulations results}, followed by conclusions and future directions in Section \ref{conlcuding remarks}.


\vspace{0pt}
\section{Modeling Framework} \label{modeling framework}
\vspace{0pt}
\begin{figure*}
    \centering
    \begin{subfigure}[b]{0.72\textwidth}
        \includegraphics[width=\textwidth]{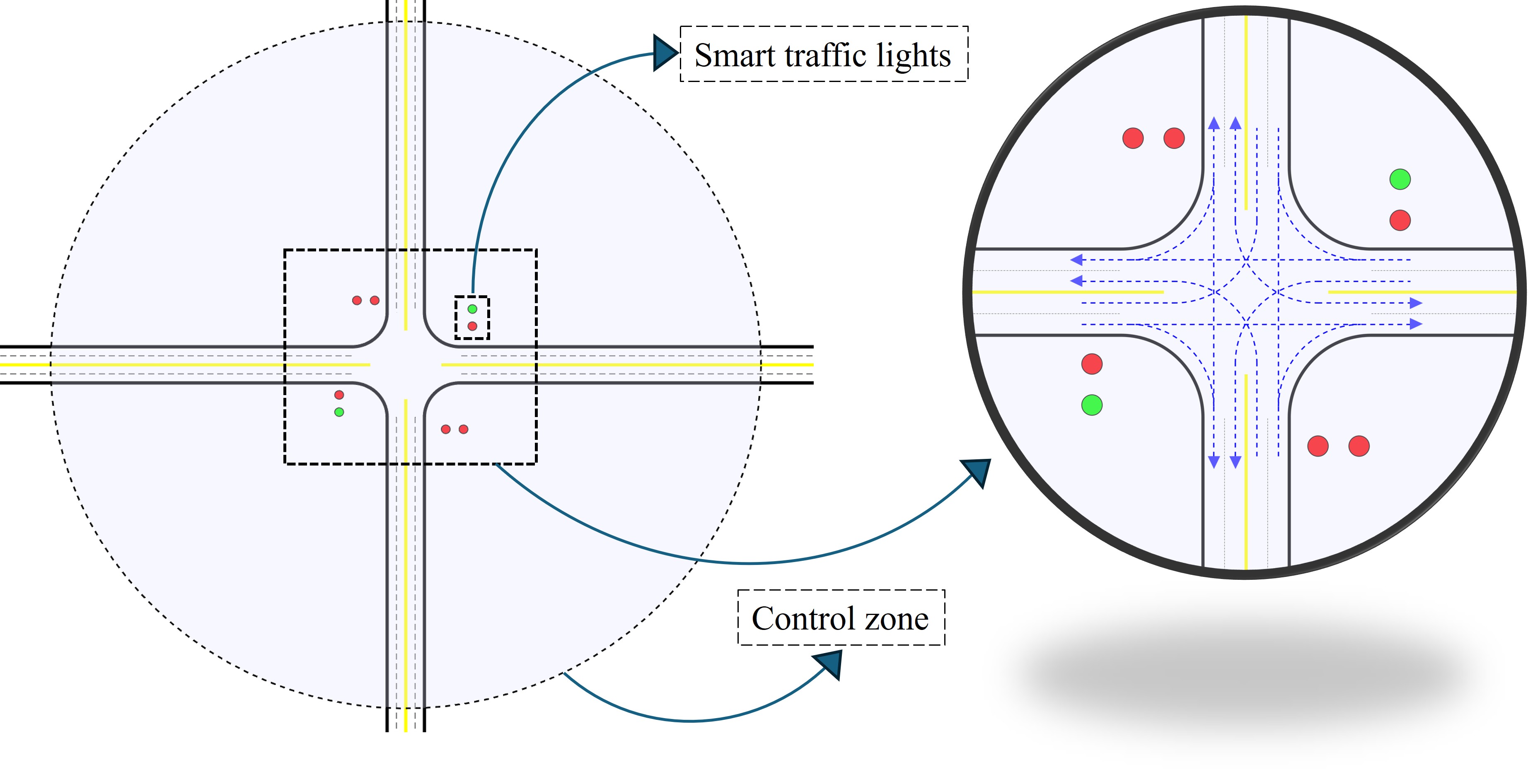}
        \caption{Illustration of a signalized intersection along with the control zone and smart traffic lights.}
        \label{fig:Intersection}
    \end{subfigure}%
    \hspace{5mm} 
        \begin{subfigure}[b]{0.24\textwidth}
        \vspace{15pt}
        \includegraphics[width=\textwidth]{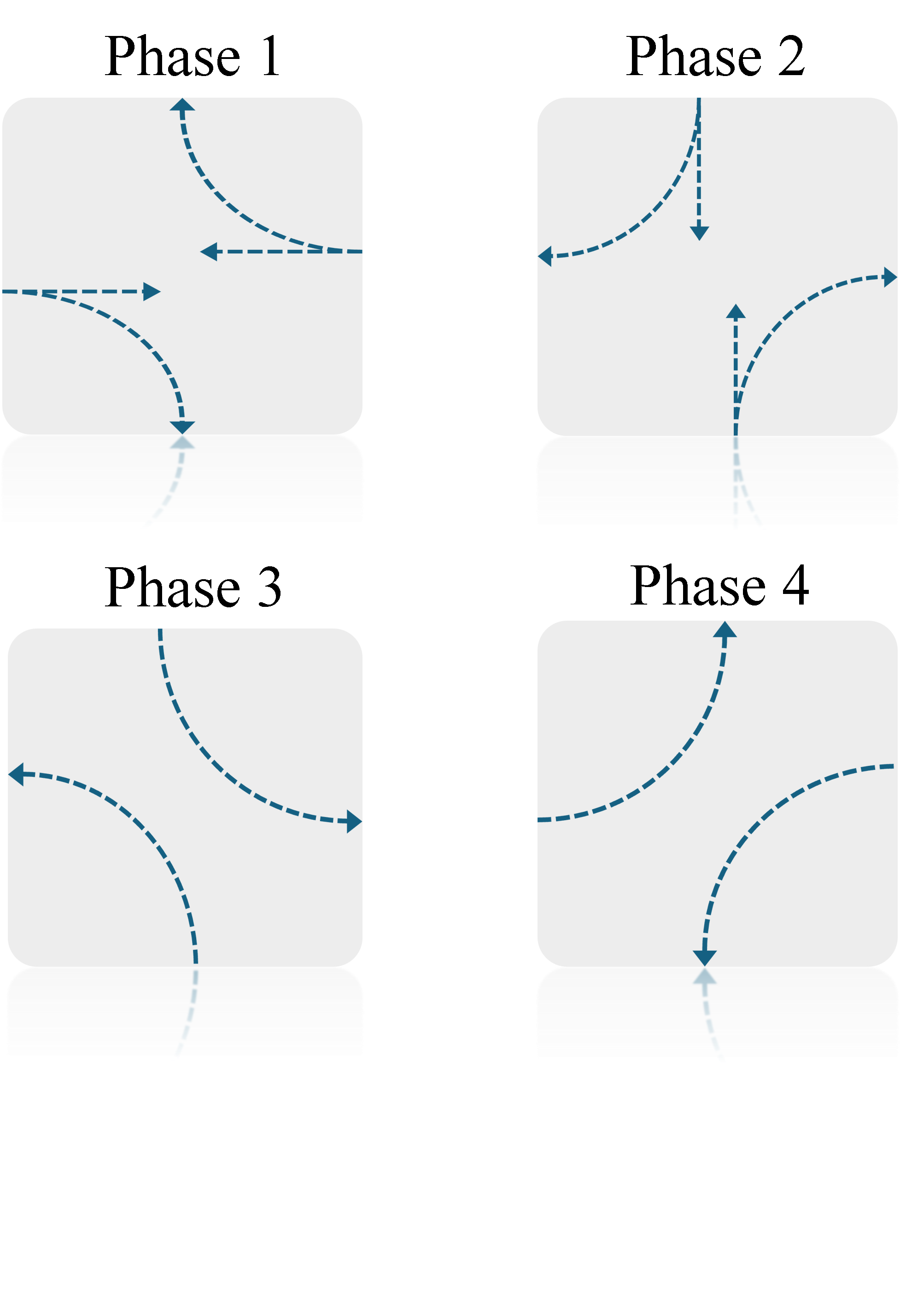}
        \caption{Four traffic phases guaranteeing no lateral collision. }
        \label{fig:phases}
    \end{subfigure}%
    \caption{Illustration of intersection set-up and traffic signal phases.}
      \vspace{-14pt}
    \end{figure*}

We analyze an intersection as shown in Fig. \ref{fig:Intersection} since it consists of one of the most challenging traffic scenarios. Let the set of paths in the intersection be denoted as $\mathcal{G}=\{1,\dots,|\mathcal{G}|\}$; see the blue arrows in Fig. \ref{fig:Intersection}.  We also define as \textit{control zone} the area within CAVs can exchange information with each other and with the \textit{smart traffic lights}. The smart traffic lights are capable of storing information related to the vehicles and adjusting their phase duration based on this information as we will discuss in Section~\ref{traffic lights control}.  Consider now a collection of vehicles composed of \(N\) CAVs and \(M\) HDVs. The set of CAVs within the control zone at time $t$ is denoted as $\mathcal{N}(t)=\{1,\dots,N(t)\}$ while the set of HDVs is denoted as $\mathcal{M}(t)=\{N+1,\dots,M(t)\}$.  Each vehicle's state is described by its longitudinal position \(p_i\) and velocity \(v_i\), along with a control input \(u_i\) that represents acceleration. Regarding $u_i$, CAVs follow a predetermined control policy, while HDVs are controlled by humans. These dynamics are described by a double-integrator model,
\begin{equation}  
\begin{aligned} \label{eq:dynamics}
\dot{p}_i = v_i, \\ \quad \dot{v}_i = u_i,
\end{aligned}
\end{equation}
where $i\in\mathcal{N}(t)\cup \mathcal{M}(t).$
As depicted in Fig. \ref{fig:Intersection}, the intersection is controlled by traffic lights whose phases allow vehicles to cross if there is no risk of a lateral collision. To simplify our analysis, we define only four traffic phases (see Fig. \ref{fig:phases}); yet, adjusting the number of phases according to CAVs penetration rate, is the subject of ongoing research. Let us now define the following assumptions.

\begin{itemize}
    \item Assumption 1: Tracking controllers are available at the low-level control layer of CAVs to ensure reference speed tracking and steering for lane keeping. 
    \item Assumption 2: Communication noise and delay between CAVs and traffic lights is negligible.
    \item Assumption 3: Vehicles have already selected their path upon entering the control zone.
    \item Assumption 4: HDVs follow a car-following model and based on this model their future states can be predicted.
\end{itemize}

We impose Assumption 1 to focus explicitly on higher-level decision-making, though deviations could easily be incorporated if needed; see \cite{ChalakiCBF2022}. Assumption 2 can be similarly relaxed by introducing communication noise, as demonstrated in previous work \cite{chalaki2021RobustGP}. Assumption 3 is reasonable given the proximity of the control zone to the intersection. Although studies have considered relaxing this assumption \cite{Malikopoulos2020}; doing so does not affect our analysis and thus we exclude it to focus on our main results. Finally, while Assumption~4 may seem restrictive, analyzing HDV behavior in a signalized environment requires extensive data analysis, which is the subject of ongoing research, and due to space limitations is out of the scope of this work. Nonetheless, in Section \ref{Human_Driven_vehicles_section} we discuss how we can safely handle possible deviations of the actual trajectory from the predicted one.

\section{Control of Connected and Automoated Vehicles} \label{CAVs control}
In this section, we define a control policy for CAVs based on signal timing and preceding CAVs (or trajectory predictions for HDVs), with the aim of increasing throughput and minimizing energy consumption. Next, we introduce the following definitions that will facilitate our exposition. Let $t_i^0$ denote the time at which CAV $i$ enters the control zone, $t_i^{\text{tr}}$ the time at which it crosses the traffic light, and $t_i^f$ the time at which it exits the control zone. Similarly, let $p_i^0$, $p_i^{\text{tr}}$, and $p_i^f$ represent the corresponding positions at these times, respectively. Next, the set of available green time intervals related to the path of vehicle $i$ is denoted by $\mathcal{T}_i^g$. Only for this section, we consider that the duration of the signal phases is known while in Section IV we present a method to adjust the phase duration without interfering with existing CAVs' plans, with the aim of improving throughput. Following, we introduce the safety constraints for CAVs.

\subsubsection{Rear-end constraints}
To prevent any possible rear-end collision between two consecutive vehicles, $i \in \mathcal{N}$ and $k \in (\mathcal{N} \cup \mathcal{M}) \setminus {i}$, where $k$ is the preceding vehicle, we impose the following constraint
\begin{equation}\label{eq:rear-end}
    p_k(t)-p_i(t)\geq \delta_i(t) = \phi v_i(t) + \gamma_{k},
\end{equation}
\noindent where $\delta_i(t)$ is the safe-dependent distance, while $\phi$ and $\gamma_k$ are the reaction time and standstill distance, respectively. We use a different standstill distance per vehicle to allow higher values when the preceding vehicle is an HDV, allowing for more conservative spacing.

\subsubsection{State constraints} Each CAV must respect both minimum and maximum speed limits, as well as minimum and maximum control inputs, defined as
\begin{align}
    &v_{\text{min}}\leq v_i(t) \leq v_{\text{max}}, \label{speed constraint} \\
    &u_{\text{min}}\leq u_i(t) \leq u_{\text{max}}. \label{acceleration constraint} 
\end{align}

\subsubsection{Traffic light constraint} The time each CAV $i$ crosses the traffic light must satisfy traffic light constraints
\begin{align}
    t_i^{\text{tr}}\in\mathcal{T}_i^g. \label{traffic_Constraint}
\end{align}
\subsection{Finding an unconstrained trajectory} \label{finding unconstrained trajectory}
The goal of each CAV $i$ is to minimize its travel time while also minimizing its control input (acceleration). These objectives aim to maximize throughput at the intersection while indirectly reducing energy consumption. However, they are competing objectives, leading each CAV to a problem of balancing the associated trade-off.  To manage the tradeoff between the two objectives in an optimal way, we utilize an approach presented in \cite{Malikopoulos2020}. Namely, we formulate two optimization problems: an upper-level time-optimal problem that identifies the exit time from the control zone for each CAV $i$, which is then passed as an input to a low-level optimal control problem that yields the control input trajectory that minimizes energy consumption. The low-level  optimization problem is formulated as:

\noindent \textbf{Low-level (energy-optimal) control problem:} \label{prb:ocp-1}
\begin{equation}
\begin{aligned}
\label{eq:energy_cost}
&\underset{u_i}{\min} \quad \frac{1}{2} \int_{t^{0}_{i}}^{t_i^f} u^2_i(t) \, \mathrm{d}t, \\
&\text{subject to:}\quad \eqref{eq:dynamics},\quad \eqref{eq:rear-end}, \quad \eqref{speed constraint},\quad \eqref{acceleration constraint}, \quad \eqref{traffic_Constraint}, \quad \\
&\text{given:} \quad p_i (t_i^0) = p_i^0, \,\, v_i (t_i^0) = v_i^0, 
\,\, p_i (t_i^f) = p_i^f,
\end{aligned}
\end{equation}
where $t_i^f$ is computed by the upper-level optimization problem, discussed next. This problem can be solved analytically by utilizing the Hamiltonian analysis as presented by \cite{Malikopoulos2020}. The optimal unconstrained trajectories are 
\begin{equation}\label{eq:optimalTrajectory}
\begin{split}
u_i(t) &= 6 a_i t + 2 b_i, \\
v_i(t) &= 3 a_i t^2 + 2 b_i t + c_i, 
\\
p_i(t) &= a_i t^3 + b_i t^2 + c_i t + d_i,
\end{split}
\end{equation} 
where $a_i, b_i, c_i, d_i$ are constants of integration and can be found using the initial and terminal conditions along the optimal terminal condition $u_i(t_i^f)=0$ as discussed in \cite{chalaki2020experimental}. \\

\noindent \textbf{Upper-Level (time-optimal) problem}: \\ \\
At time $t_i^0$, when CAV $i$ enters the control zone, we define the feasible exit-time range 
as $\mathcal{F}_i(t_i^0) = [\underline{t}_i^f, \overline{t}_i^f]$. The bounds \(\underline{t}_i^f\) and \(\overline{t}_i^f\) are state-dependent and correspond to the earliest and latest exit time times achievable by the unconstrained trajectory \eqref{eq:optimalTrajectory} while satisfying the state and control bounds (see \cite{chalaki2020experimental}, page 27). Then, CAV \(i\) solves an optimization problem to find the minimum exit time \(t_i^f \in \mathcal{F}_i(t_i^0)\) that satisfies all constraints, defined as,
\begin{align}\label{Time_optimal}
&\underset{t_i^f \in \mathcal{F}_i(t_i^0)}{\min} \quad t_i^f  \\
&\text{subject to:} \quad \eqref{eq:rear-end}, \quad \eqref{speed constraint},\quad \eqref{acceleration constraint}, \quad \eqref{traffic_Constraint}, \quad \eqref{eq:optimalTrajectory}, \nonumber \\
&\text{given:} \quad p_i(t_i^0) = p_i^0, \, v_i(t_i^0) = v_i^0, \, p_i(t_i^f) = p_i^f, \, u_i(t_i^f) = 0. \nonumber 
\end{align}
The two optimization problems \eqref{eq:energy_cost}-\eqref{Time_optimal} aim to find an optimal unconstrained trajectory as illustrated in Algorithm 1. After obtaining a solution from Algorithm 1, the final trajectory for CAV $i$ is sent to the smart traffic lights, allowing the following CAVs to access this information and plan their own trajectories accordingly. Also note that CAVs are capable of sharing information with the smart traffic lights regarding HDV trajectories, so other CAV can access it, as well. 

\begin{algorithm}[H]
\caption{Finding an unconstrained trajectory}
\begin{algorithmic}[1]
\STATE Initialize $t_i^f \gets \underline{t}_i^f$
\WHILE{$t_i^f \in \mathcal{F}(t_i^0)$}
    \STATE Solve the optimization problem~\eqref{eq:energy_cost}
    \IF{no constraints are violated}
        \STATE Output: solution is~\eqref{eq:optimalTrajectory} based on $t_i^f$
        \STATE \textbf{break}
    \ELSE
        \STATE $t_i^f \gets t_i^f + \Delta t$
    \ENDIF
\ENDWHILE
\end{algorithmic}
\end{algorithm}
\subsection{Case where no feasible solution exists due to traffic light constraints}

When a CAV $i$ follows the procedure described in Algorithm~1, there may be instances where no feasible \textit{unconstrained} trajectory exists, e.g., due to traffic signal timing. To rapidly identify such cases without executing Algorithm 1, we establish a real-time condition that determines whether there exists a subset of unconstrained trajectories for CAV $i$ that respects signal timing. Only if there exists such a subset, CAV $i$ can apply Algorithm 1 to identify an unconstrained trajectory that also respects the rest of the constraints. Let us now formally define this condition.

Consider a CAV~$i$, which enters the control zone at time $t_i^0$. Based on its initial state, the feasible exit time interval from the control zone is $\mathcal{F}_i(t_i^0) = [\underline{t}_i^f, \overline{t}_i^f]$. Then, define as $\underline{p}_i(t)$ and $\overline{p}_i(t)$ the unconstrained position trajectories associated with the exit times $\underline{t}_i^f$ and $\overline{t}_i^f$, respectively. Let $t_i^{c_1}$ and $t_i^{c_2}$ be the times when these trajectories reach the traffic light location $p_i^{tr}$, that is, $\underline{p}_i(t_i^{c_1}) = p_i^{tr}$ and $\overline{p}_i(t_i^{c_2}) = p_i^{tr}$. Then, consider the green interval of the traffic light for this traffic cycle to be given by $\mathcal{T}_i^g = [g_1, g_2]$. To determine whether there is a subset of unconstrained trajectories that allows CAV $i$ to cross during $\mathcal{T}_i^g$, we compute the interval
\begin{align} \label{crossing check}
[t_i^{c_1}, t_i^{c_2}] \cap [g_1, g_2] &= \left[\max\{t_i^{c_1}, g_1\}, \min\{t_i^{c_2}, g_2\}\right] \nonumber \\ 
&=[T_{\text{start}}, T_{\text{end}}].
\end{align}
\noindent Two cases arise:  1) If $T_{\text{start}} \leq T_{\text{end}}$, there exists a subset of feasible unconstrained trajectories that satisfies the traffic light constraint. 2) If $T_{\text{start}} > T_{\text{end}}$, no feasible unconstrained trajectory exists for crossing the traffic light within the green phase. For each of the above cases, CAV $i$ must follow a different control strategy. 

In the first case, CAV $i$ proceeds by following Algorithm~1. In case Algorithm~1 does not return a feasible solution due to rear-end constraints, the CAV enters a standby mode which is defined in Section~\ref{standby mode}.

For the second case, there are two alternatives: 1) CAV $i$ identifies a feasible \textit{constrained} trajectory that activates one or both of the speed and acceleration constraints \eqref{speed constraint}--\eqref{acceleration constraint} over a non-zero time interval while ensuring the vehicle can cross the intersection in compliance with the signal timing and rear-end constraints (discussed in Section \ref{constrained trajectory}) or 2) CAV enters a standby mode defined in Section~\ref{standby mode}.



\begin{remark}
When a CAV enters the control zone, it may have access to a disjoint set of green intervals. In such cases, it must sequentially check the previous condition for all intervals before selecting to enter the standby mode.
\end{remark}

In the following subsections, we present the strategy for adopting a constrained trajectory and then the strategy for entering a standby mode when no feasible trajectory exists.

\subsection{Finding a constrained trajectory} \label{constrained trajectory}

\noindent As previously discussed, we prioritize unconstrained trajectories to maintain a balance between energy consumption and travel time. However, if for a CAV $i$ we obtain $T_{\text{start}} > T_{\text{end}}$, indicates that the vehicle cannot cross the intersection within the allowed time interval using an unconstrained trajectory. However, this infeasibility does not necessarily imply that crossing is impossible. The vehicle might still cross in time by following a constrained trajectory, by applying maximum acceleration $u_{\max}$ or (and) reaching maximum speed $v_{\max}$, depending on its initial state.

To identify if a constrained solution can be considered, we first need to identify the earliest time a CAV $i$ can reach the traffic light at $p_i^{tr}$ based on its physical capabilities. Depending on whether CAV $i$ reaches $v_{\max}$ before or after the traffic light under maximum acceleration $u_{\max}$, the minimal achievable arrival time $T_i^{\mathrm{min}}(v_0)$ at $p_i^{tr}$ is given by,
\begin{equation}
T_i^{{\mathrm{min}}}
= 
\begin{cases}
\frac{\sqrt{\,(v_i^0)^2 + 2\,u_{\max}\,p_i^{tr}\,} - v_i^0}{u_{\max}}, 
& p^{v_{\max}} \ge p_i^{tr}, \\[1em]
\frac{v_{\max} - v_i^0}{u_{\max}}
+ 
\frac{p_i^{tr} - p^{v_{\max}}}{\,v_{\max}\,}
&p^{v_{\max}} < p_i^{tr},
\end{cases}
\end{equation}
where $p^{v_{\max}}=\frac{v_{\max}^2 - (v_i^0)^2}{2 u_{\max}}$   is the position at which the vehicle reaches the maximum allowed speed $v_{\max}$, under maximum acceleration $u_{\max}$. As an example, we can see in Fig. \ref{fig:Example1}, at \textit{Instant 1} that although the interval associated with the unconstrained trajectories $[t_i^{c_1}, t_i^{c_2}]$ lies within the red phase, if the vehicle accelerates at its maximum it may be able to reach the green light at $T_i^{{\mathrm{min}}}$, by following a constrained trajectory. 

The construction of a constrained trajectory to achieve a faster arrival time $t_i^{tr}$ at the traffic light is a rapid and straightforward process, detailed in \cite[Section III.B, \textit{Problem} 3]{meng2020eco}. Due to space limitations, we omit it here.


\begin{remark}
Note that when defining constrained trajectories, we do not directly incorporate rear-end constraints 
because that would require a time-intensive and complicated process of merging constrained and 
unconstrained arcs based on the dynamic state evolution of the preceding vehicle. Following, we discuss how to select $t_i^{tr}$ aiming for a constrained trajectory that does not violate rear-end constraints.
\end{remark}

To identify a feasible $t_i^{tr}$, we first determine the interval in which 
$t_i^{tr}$ can lie, namely $[\mathcal{T}_i^{\min}, g_2]$, where $g_2$ is the 
last available crossing time (see Fig.~\ref{fig:Example1}, \emph{Instant 1}). 
If $[\mathcal{T}_i^{\min}, g_2]=~\emptyset$, the CAV enters standby mode since it cannot cross due to its physical capabilities. Otherwise, we check if a preceding vehicle $k$ occupies a subset of 
$[\mathcal{T}_i^{\min}, g_2]$. In this case, that interval reduces to 
$[t_k^{tr} + \epsilon, g_2]$, with $\mathcal{T}_i^{\min}<t_k^{tr} + \epsilon$ 
where $\epsilon$ is the time headway calculated from~\eqref{eq:rear-end}. Let $\mathcal{C}(t_i^0)$ denote the resulting interval, either $[\mathcal{T}_i^{\min}, g_2]$ or $[t_k^{tr} + \epsilon, g_2]$.

Based on $\mathcal{C}(t_i^0)$ we initialize $t_i^{tr} \gets \min(\mathcal{C}(t_i^{tr}))$. Given that $\mathcal{C}(t_i^0)$ is only associated with the crossing time, we also check whether the constrained trajectory associated with $t_i^{tr}$ violates the rear-end constraint \eqref{eq:rear-end} from $p_i^0$ to $p_i^{tr}$. If \eqref{eq:rear-end} is violated, we increase $t_i^{tr}$ by an incremental step $\Delta t$ within $\mathcal{C}(t_i^0)$ and 
repeat this check.   

Finally, we also need to verify whether an unconstrained trajectory from $p_i^{tr}$ to $p_i^{f}$—that is, from the traffic light to the exit of the control zone—is feasible for CAV~$i$ based on the selected time~$t_i^{tr}$. Based on the state of CAV $i$ at $t_i^{tr}$ (which is known upon selecting  $t_i^{tr}$) we determine the feasible exit-time interval from the control zone starting from $t_i^{tr}$, that is,
$\mathcal{F}_i(t_i^{tr}) = [\underline{t}_i^f, \overline{t}_i^f]$ (similarly to $\mathcal{F}_i(t_i^{0})$ discussed above). Then, we verify whether the unconstrained trajectory with the largest exit time 
$\overline{t}_i^f$ satisfies the rear-end safety 
constraint \eqref{eq:rear-end}. If that trajectory is infeasible, all 
trajectories with smaller exit times $t_i^{tr}\in \mathcal{F}_i(t_i^{tr})$ are also infeasible due to higher velocities. In such case, we increase again $t_i^{tr}$ by an incremental step $\Delta t$ within $\mathcal{C}(t_i^0)$ and 
repeat this check. If the new unconstrained trajectory from $p_i^{tr}$ to 
$p_i^f$ is feasible, CAV~$i$ can adopt a constrained trajectory to 
$p_i^{tr}$, ensuring it can safely exit the control zone after crossing. The previous process within $\mathcal{C}(t_i^0)$ can be done in a fraction of a second, given that we have access to the other vehicles' trajectories, from the smart traffic lights.

Finally, if $t_i^{tr}$ is identified through the previous process, then CAV $i$ may use Algorithm~1 to choose 
an optimal unconstrained trajectory from $p_i^{tr}$ to $p_i^f$ with a 
potentially lower feasible exit time $t_i^f$ within 
$\mathcal{F}_i(t_i^{tr})$. In case $t_i^{tr}$ is not identified, CAV $i$ enters the standby mode discussed next.

\begin{figure}
    \centering
    \includegraphics[width=1\linewidth]{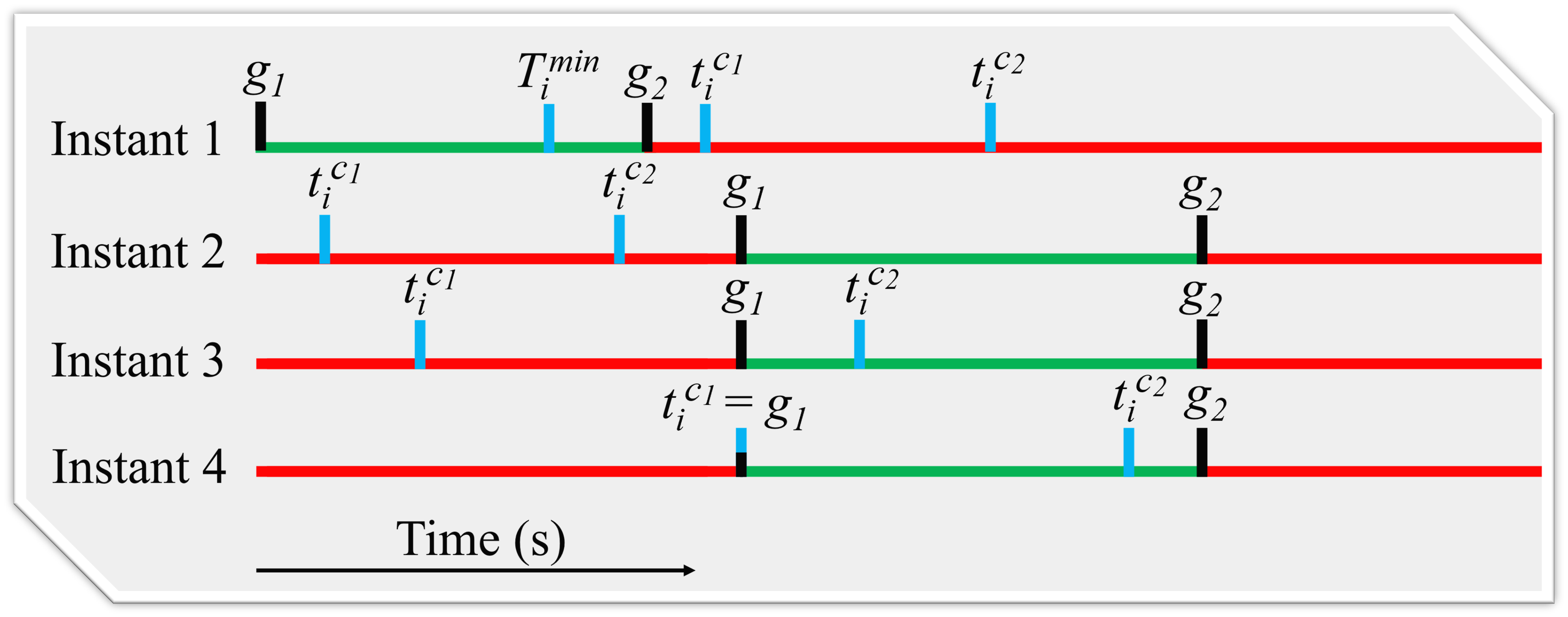}
    \caption{Illustration of different instants with varying crossing intervals and times $\mathcal{T}_i^{\min}$, $t_i^{c_1}$, and $t_i^{c_2}.$}
    \label{fig:Example1}
      \vspace{-12pt}
\end{figure}

\subsection{Standby mode} \label{standby mode}
When the CAV’s physical capabilities and rear-end constraints prevent it from safely crossing the traffic light after evaluating the methods discussed above, the CAV enters a standby mode, awaiting for the next available crossing interval. During the standby mode the CAV follows an optimal unconstrained trajectory to stop at the traffic light. To define this trajectory, we employ the framework presented in \eqref{eq:energy_cost}, resulting in the unconstrained trajectory form given by \eqref{eq:optimalTrajectory}.  To determine the constants of integration, we use the initial position $p_i(t_i^0)=p_i^0$,  initial velocity $v_i(t_i^0)=v_i^0$, terminal position at the traffic light $p_i(t_i^{tr})=p_i^{tr}$ and the terminal velocity $v_i(t_i^{tr})=0$. Finally in order to get the analytical solution form \eqref{eq:optimalTrajectory} we need to identify the stopping time $t_i^{tr}$ denoted for clarity as $t_i^s$. We aim to select the \textit{latest} feasible stopping time $t_i^s$ that allows for an unconstrained trajectory, for two reasons: (1) an earlier stop increases energy consumption due to the trade-off between energy use and travel time, and (2) this approach provides the CAV with additional time in standby mode to potentially identify an unconstrained trajectory to cross $p_i^{tr}$ without stopping. We identify the latest stopping time in the following Theorem. 

\begin{theorem}
For CAV $i$, the latest stopping time  $t_i^s$ to stop at $p_i^{tr}$ under an unconstrained trajectory, is defined as:
\begin{align}
t_i^s &=
\left\{
\begin{array}{ll}
-\frac{2v_i^0}{u_i^c}, & t_b\leq u_i^c,  \\[5pt]
t_b, & t_b> u_i^c,
\end{array}
\right.  
\end{align}
where $t_b=\frac{-v_i^0+\sqrt{(v_i^0)^2-6u_{\min}p_i^{tr}}}{-u_{\min}}$ and $u_i^c=-\frac{2(v_i^0)^2}{3p_i^{tr}}$ .
\end{theorem}

\begin{proof}
Without loss of generality, consider initial conditions $p_i(t_i^0)=0$, $v_i(t_i^0)=v_i^0$. The final conditions at the traffic light are $p_i(t_i^{tr})=p_i^{tr}$ and $v_i(t_i^{tr})=0$, where the terminal time $t_i^{tr}$ is a free variable. From \eqref{eq:optimalTrajectory}, we obtain:
\begin{align}
&a_i (t_i^{tr})^3 + b_i (t_i^{tr})^2 + v_i^0 t_i^{tr} - p_i^{tr} = 0, \label{first}\\
&3 a_i (t_i^{tr})^2 + 2 b_i t_i^{tr} + v_i^0 = 0. \label{second}
\end{align}
From \eqref{second}, we solve for $a_i$ and obtain $a_i = \frac{-2 b_i t_i^{tr}-v_i^0}{3 (t_i^{tr})^2}$. Thus, substituting $a_i$ into \eqref{first} yields $\frac{-2 b_i t_i^{tr}-v_i^0}{3 (t_i^{tr})^2}(t_i^{tr})^3 + b_i (t_i^{tr})^2 + v_i^0 t_i^{tr} - p_i^{tr}$. Simplifying terms yields
\begin{align}
\frac{1}{3} b_i(t_i^{tr})^2 +\frac{2}{3}v_i^0t_i^{tr} - p_i^{tr} &= 0. \label{third}
\end{align}

\noindent Let the initial acceleration be denoted by $u_i^0(t_i^0) = u_i^0$. From equation~\eqref{eq:optimalTrajectory}, we obtain
$u_i^0 = 2b_i \Rightarrow b_i = \frac{u_i^0}{2}.$
Substituting this into equation~\eqref{third}, yields the following quadratic equation in $t_i^{tr}$:
\begin{align}
\frac{u_i^0}{6}(t_i^{tr})^2 + \frac{2v_i^0}{3}t_i^{tr} - p_i^{tr} = 0. \label{quadratic}
\end{align}

For equation~\eqref{quadratic} to admit real solutions, the discriminant must be non-negative. The discriminant $\Delta$ of~\eqref{quadratic} is given by: $\Delta = \left(\frac{2v_i^0}{3}\right)^2 + \frac{2u_i^0}{6}p_i^{tr} = \frac{4(v_i^0)^2 + 6u_i^0p_i^{tr}}{9}.$ Setting $\Delta=0$  allows us to find the critical initial acceleration $u_i^c$. Thus we get $\Delta = 0 \Rightarrow  4(v_i^0)^2 + 6u_i^c p_i^{tr} =~0 \Rightarrow u_i^c = -\frac{2(v_i^0)^2}{3p_i^{tr}}.$ Then for $\Delta = 0$ the solution to \eqref{quadratic} is trivial and equal to
\[
t_i^s = -\frac{2v_i^0}{u_i^c}.
\]
\noindent If $u_i^0 < u_i^c$, then $\Delta<0$, leads to complex roots, which implies the vehicle stops before reaching the traffic light. 



Next, to guarantee that $u_{\min}$ is satisfied, we need to know the sign of $a_i$. Let us consider first that $u(t)$ is decreasing. Then, $a_i<~0$. Thus, it is sufficient to prove that $u_{\min}$ is not active at $t_i^s$. Hence, we aim to identify the stopping time $t_i^s$ that drives CAV $i$ to activate $u_{\min}$ at $t_i^s$. From \eqref{eq:optimalTrajectory} we obtain $6a_i(t_i^{tr})+2b_i=u_{\min}$ and we know $a_i=\frac{-2b_i\,t_i^{tr}-v_i^0}{3(t_i^{tr})^2}$. Then we solve for the parameter $b_i$. Substituting $b_i$ into \eqref{third} leads to a quadratic equation in $t_i^{tr}$. The positive root is given by $t_b=\frac{-v_i^0+\sqrt{(v_i^0)^2-6u_{\min}p_i^{tr}}}{-u_{\min}}$, which represents the critical time that makes $u_{\min}$ be activated at $t_i^s$. Thus, we need $t_i^s\geq t_b$. Otherwise $t_i^s=t_i^b$. That completes the proof.
\end{proof}

\begin{remark}
A similar argument applies for increasing $u_i(t)$. Then, $u_{\min}$ is satisfied if $u_{\min}$ is not activated at $t_i^0$. However, even with confined distances $p_i^{tr}$ (e.g., $p_i^{tr}=40$ m) and unrealistic initial speeds for urban environments (e.g., $v_i^0=25$ m/s), the critical acceleration $u_i^c = -\frac{2(v_i^0)^2}{3p_i^{tr}} = -5.21$ m/s\textsuperscript{2} yields realistic decelerations; so the proof is omitted.
\end{remark}

Given that we defined the stopping time $t_i^s$, we can now trivially determine an unconstrained standby trajectory for each CAV \( i \) by applying \eqref{eq:energy_cost}-\eqref{eq:optimalTrajectory}. Note that if other preceding vehicles in the same lane are also in standby mode, the final position \( p_i^{tr} \) should be adjusted to account for the preceding vehicle by incorporating the standstill distance.

\subsection{Event-triggered re-planning} \label{event-triggered replanning}
During the standby trajectory, the CAV continuously updates in real time the values $T_{\text{start}}(t)$ and $T_{\text{end}}(t)$ as defined in \eqref{crossing check}, based on its state at time $t$. Thus if the following holds
\begin{equation}
 [T_{\text{start}}(t), T_{\text{end}}(t)] \neq \emptyset\quad t\;\in [t_i^0, t_i^s],  
\end{equation}
for $t>t_i^0$ the CAV can apply Algorithm 1 to identify an unconstrained trajectory to cross the traffic light and exit the standby mode. For instance, consider a CAV arriving at the control zone at time \(t_i^0\) with an initial crossing interval \([t_i^{c_1}, t_i^{c_2}]\) (Fig.~\ref{fig:Example1}, Instant 2). Because it cannot cross immediately, it follows a standby trajectory. During its deceleration phase (Instant 3), the CAV updates its interval such that $g_1 < t_i^{c_2}$ resulting in $[T_{\text{start}}(t), T_{\text{end}}(t)] \neq \emptyset$, enabling an unconstrained trajectory that passes through the green phase.

Finally, note that a CAV exiting the standby mode upon identifying an unconstrained trajectory may still need to replan at a subsequent stage where $t_i^{c_1}$ satisfies $ t_i^{c_1} \ge g_1$, to adopt an unconstrained trajectory, without getting influenced by the signal constraint. For instance, in Fig.~\ref{fig:Example1} (Instant 3), although \( t_i^{c_2} > g_1 \) allows for an unconstrained trajectory, it still accounts for the red phase, resulting in a slower exit \( t_i^f \). Once \( t_i^{c_1} = g_1 \) (Instant 4), the traffic light no longer imposes restrictions, and replanning can yield a faster exit time $t_i^f.$

\vspace{-3pt}
\section{Human driver model} 
\vspace{-3pt}
\label{Human_Driven_vehicles_section}
While various car-following models have been proposed and substantial work exists on human driver prediction, to the best of our knowledge, no existing work accurately predicts human driver intentions at signalized intersections with adaptive traffic signals under mixed traffic conditions. In previous work \cite{le2024stochastic}, we employed Bayesian linear regression to efficiently learn human behavior from online data and predict HDV trajectories in real time in merging scenarios. However, in signalized intersections with adaptive traffic signals and CAVs, driver behavior prediction remains challenging due to lack of available data. As discussed in Section \ref{modeling framework}, this is the focus of ongoing research utilizing Virtual Reality to aggregate sufficient data under such conditions.
 
Here, we adopt a simple yet easily extendable model (from the perspective of parameterization) to account for future HDVs trajectories, the Intelligent Driver Model (IDM). The acceleration $u_i$ of a vehicle based on IDM is given by:
\begin{equation} \label{IDM}
u_i(t) = u_{\max} \left[ 1 - \left( \frac{v_i(t)}{v_{des}} \right)^{\delta} - \left( \frac{s^*(v_i(t), \Delta v)}{s_i(t)} \right)^2 \right],
\end{equation}
where, $v(t)$ is the current speed of the vehicle, $v_{des}$ is the desired speed, $u_{\max}$ is the maximum acceleration, $\delta$ is a model parameter, $s(t)$ is the gap to the leading vehicle and $\Delta v = v(t) - v_{\text{lead}}(t)$ is the speed difference between the vehicle and the vehicle ahead. The desired minimum gap \( s^*(v_i(t), \Delta v) \) is given by: $s^*(v_i(t), \Delta v) = \gamma_i + v_i(t) \tau + \frac{v_i(t) \Delta v}{2 \sqrt{u_{\max} \beta}}$ where $\gamma_i$ is the standstill distance, $\tau$ is the desired time headway and $\beta$ is the comfortable deceleration.

To incorporate traffic lights into the IDM, we introduce a stationary virtual vehicle at the position of the red light, allowing the HDV to adjust its acceleration based on the distance to the traffic light considering $\Delta v=0$. When both surrounding traffic and red light are present, the IDM accordingly selects the lower acceleration from the red traffic light and the preceding vehicle.

\begin{remark}
Given that we use the IDM model in this paper, predicting the trajectories of HDVs is straightforward and can be computed in milliseconds using efficient ordinary differential equation solvers while considering the evolution of the whole system. Interested readers can experiment with our code provided on the paper's website (see Section \ref{simulations results}). Ongoing research focuses on data-driven methods for HDV trajectory prediction under uncertainty. 
Then, whenever an HDV’s actual path deviates beyond a defined threshold, a replanning mechanism is triggered, 
allowing the following CAVs to adjust their trajectories accordingly.
\end{remark}


\section{Traffic Lights Control} \label{traffic lights control}
In this section, we present our traffic light signal policy. Let $\mathcal{L} = \{1, \dots, |\mathcal{L}|\}$ denote the set of traffic lights, and $\mathcal{P} = \{1, \dots, |\mathcal{P}|\}$ the set of traffic phases.  For each traffic light $l \in \mathcal{L}$, we define $\mathcal{Z}_l$ as the set of paths included in $\mathcal{G}$ controlled by that light, considering that the same light can govern two different paths; see Fig. \ref{fig:Intersection}. Similarly, for each phase $p \in \mathcal{P}$, let $\mathcal{L}_p$ represent the set of traffic lights that are green during that phase. Define the duration of phase $p\in\mathcal{P}$ as $d_p$ and let $T_{\text{cycle}}$ denote the total cycle duration, defined as the sum of the durations of the four phases, namely $T_{\text{cycle}}=\sum_{p \in \mathcal{P}} d_p$ . Let us now define the \textit{pressure} of phase $p$, denoted by $P_p$, as follows:
\[
P_p = \sum_{l \in \mathcal{L}_p} \sum_{z \in \mathcal{Z}_l} n_z \quad p\in \mathcal{P},
\]
where $n_z$ denotes the number of CAVs in \textit{standby mode}, along with the number of HDVs that are stopped and are predicted to stop on path $z$ whose traffic light is associated with phase $p$. This metric evaluates the pressure by exclusively considering the vehicles contributing congestion.

Our final goal is to minimize the pressure of each phase, without influencing the trajectories of the CAVs that have already planned their trajectories. To achieve that we define an update policy that determines the signal timing of the subsequent traffic cycle. We define as $T_{\text{update}}$ the time when the determination takes place at the current cycle while at the same time the signal timing for the next cycle is broadcasted to CAVs. This design ensures that vehicles planning to cross during the current cycle are not affected by the update, as the update takes effect only in the following cycle. Also, it guarantees that the plan for the next cycle is known in advance, allowing CAVs in standby mode to transition to a crossing trajectory.

We now develop the framework that updates the phase durations of every upcoming traffic cycle. Let us define the vector $q = [P_1, \dots, P_{|\mathcal{P}|}]$ that contains the pressures of each phase at $T_{\text{update}}$. As a first step, we sort $q$ in descending order. Thus, let $\pi: \{1,\dots,|\mathcal{P}|\} \rightarrow \mathcal{P}$ be a permutation such that
\[
P_{\pi(1)} \ge P_{\pi(2)} \ge \dots \ge P_{\pi(|\mathcal{P}|)}, \quad \pi(\cdot
) \in \mathcal{P}.
\]
Based on this ordering, phase $\pi(1)$ will be the first phase in the next cycle, $\pi(2)$ the second, and so on. In order to guarantee that each phase is processed in every traffic cycle, we enforce a minimum duration per phase denoted as $T_{\text{min}}$. The remaining time is then calculated as $T_{\text{remain}} = T_{\text{cycle}} - 4T_{\text{min}}$.
We distribute $T_{\text{remain}}$ proportionally as \begin{align}
  T_{\pi(p),(\text{prop})} = \frac{P_{\pi(p)}}{\sum_{p=1}^{|\mathcal{P}|} P_{\pi(p)}}\, T_{\text{remain}}, \quad p \in \mathcal{P}.  
\end{align}
Thus, the final durations of the next cycle are defined as
\begin{align}  
d_p = T_{\text{min}} +  T_{\pi(p),(\text{prop})}\quad p\in\mathcal{P}.
\end{align}


\section{Simulation Results} \label{simulations results}
To validate our framework, we conducted simulations in Matlab. We selected an intersection and corresponding traffic phases as in Fig. \ref{fig:Intersection}, with a control zone of 300m range with the following parameters: $v_{max}=20$ m/s, $v_{min}=0$ m/s, $u_{\max}=5$m/$\text{s}^2$, $u_{\min}=-5$ m/$\text{s}^2$, $\beta=-2$m/$\text{s}^2$, $\tau=1.5$ s, $v_{des}=15$m/s, $T_{\text{update}}=\frac{T_{\text{cycle}}}{2}$.

Under a random seed of initial conditions, we present in Fig.~\ref{fig:result3}, the trajectories of both CAVs and HDVs over a time horizon of 80 s, with arbitrary signal timing for the first traffic cycle and $T_\text{cycle}=40$ s. 
We observe that the first vehicle (CAV) entering the control zone at \( t = 2\,\mathrm{s} \) is unable to identify an unconstrained trajectory. Consequently, it attempts and manages to identify a constrained trajectory.
Subsequently, two HDVs come to a complete stop, followed by four CAVs entering the standby mode due to their inability to identify either an unconstrained or a constrained trajectory. 
At \( t = 20\,\mathrm{s} \), the smart traffic light updates the signal timing for the next cycle, setting the vehicles' crossing time to begin at \( t = 46\,\mathrm{s} \). While approaching in standby mode, the CAVs eventually compute unconstrained trajectories that allow them to cross, as we discussed in Section \ref{event-triggered replanning}.

Similarly, in Fig.~\ref{fig:result2}, we observe a case where the first vehicle (CAV) fails to identify either a constrained or unconstrained trajectory and enters standby mode. Also the replanning instances are highlighted in the same Figure. Notably, at \( t = 30\,\mathrm{s} \), the CAV identifies an unconstrained trajectory using the event-triggered mechanism from Section~\ref{event-triggered replanning}. Then, just before crossing, it identifies a faster unconstrained trajectory with a smaller \( t_i^f \). Fig. \ref{fig:result2} also illustrates the importance of a CAV leading an HDV. Especially we see that the CAV's unconstrained trajectory indirectly guides the HDV and prevents it from making a hard stop (dotted trajectory) at the traffic light, avoiding unnecessary energy waste.

The same behaviors are consistently observed in Figs.~\ref{fig:result8} and \ref{fig:result9}. 
In Fig.~\ref{fig:result8}, we highlight the moment when a CAV transitions from standby mode to an unconstrained trajectory. 
In Fig.~\ref{fig:result9}, we examine two different simulations with identical initial conditions—one consisting of CAVs and the other of HDVs. 
The difference in their resulting trajectories showcases that the standby mode enables CAVs to execute smoother decelerations, thereby avoiding hard stops and reducing unnecessary energy consumption.

Finally, to evaluate the effectiveness of the proposed traffic control policy, we conducted a qualitative study assessing the average time required for vehicles to exit the control zone. 
This assessment was performed for 200 vehicles, under identical initial conditions, using four different control policies and three different CAV penetration rates. 
The results, summarized in Table~\ref{tab:my_table}, highlight several important trends.

One key observation is that the cycle time has a substantial impact on system performance, and its selection must account for the level of CAV penetration. 
Specifically, shorter cycle times—which correspond to more frequent updates of traffic signal phases—tend to favor higher CAV penetration rates. 
For example, with a cycle time of \( T_{\text{cycle}} = 20\,\mathrm{s} \) and a CAV penetration rate of 70\%, the average exit time was 47.87 s. 
In contrast, with the same penetration rate but a longer cycle time of \( T_{\text{cycle}} = 40\,\mathrm{s} \), the average exit time increased to 67.26 s. 
This suggests that infrequent updates can hinder CAV performance.
A similar pattern is observed for a 50\% penetration rate. 
Notably, at 0\% penetration (i.e., only HDVs), the shortest cycle time \( T_{\text{cycle}} = 20\,\mathrm{s} \) resulted in significantly higher average times. 
This can be attributed to the inability of HDVs to anticipate future signal phases, leading to inefficient, one-by-one crossings due to human reaction delays. 
Under 0\% CAVs, the best performance was achieved with \( T_{\text{cycle}} = 30\,\mathrm{s} \), indicating that HDVs benefited from longer cycle durations, which allowed more time to react. It is noteworthy, however, that for \( T_{\text{cycle}} = 40\,\mathrm{s} \), under fixed signal timing, HDVs obtained significantly lower average exit time, whereas the adaptive policy yielded a higher average of 70.35 seconds. Thus, we conclude that at 0\% CAVs penetration, the choice of \( T_{\text{cycle}} \) is substantially sensitive to the performance of the system under this traffic control policy. 
Finally, it is noticeable that when \( T_{\text{cycle}} = 40\,\mathrm{s} \) under adaptive cycles, all vehicles experience their worst performance, suggesting that infrequent updates are inadequate for minimizing the average exit time.

Additional results, videos, and code are available at: \href{https://sites.google.com/cornell.edu/mixed-traffic-inter}{https://sites.google.com/cornell.edu/mixed-traffic-inter}.

\begin{figure*}[ht]
    \centering
    \begin{subfigure}[b]{0.66\linewidth}
        \includegraphics[width=\linewidth]{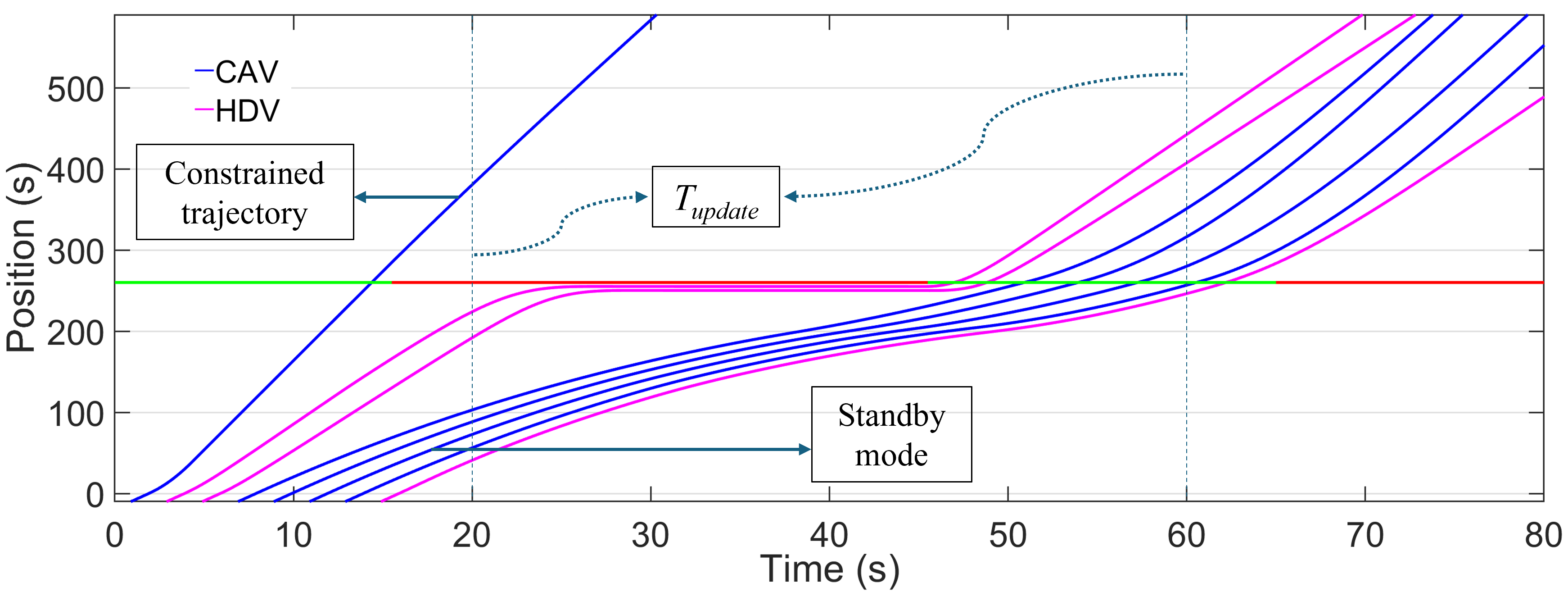}
        \vspace{-16pt}
        \caption{An example where CAVs obtain both constrained and unconstrained trajectories.}
        \label{fig:result3}
    \end{subfigure}
    \hfill
    \begin{subfigure}[b]{0.33\linewidth}
        \includegraphics[width=\linewidth]{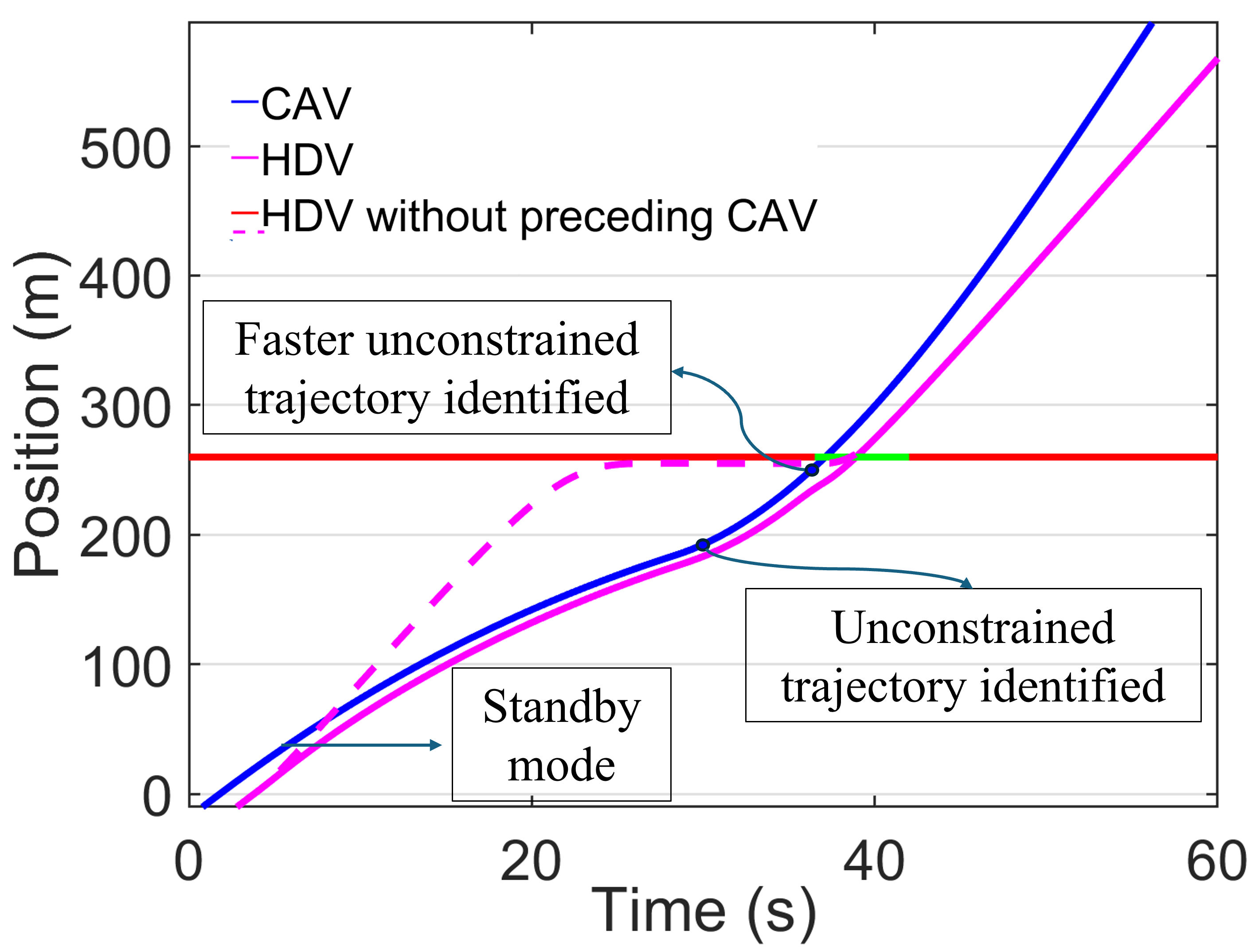}
                \vspace{-16pt}
        \caption{HDV follows CAV in standby mode.}
        \label{fig:result2}
    \end{subfigure}
    \caption{CAVs and HDVs crossing the traffic light under different traffic cycles.}
    \label{fig:overall}
        \vspace{-14pt}
\end{figure*}
\begin{figure}[ht]
    \centering
    \begin{subfigure}[b]{0.49\linewidth}
        \includegraphics[width=\linewidth]{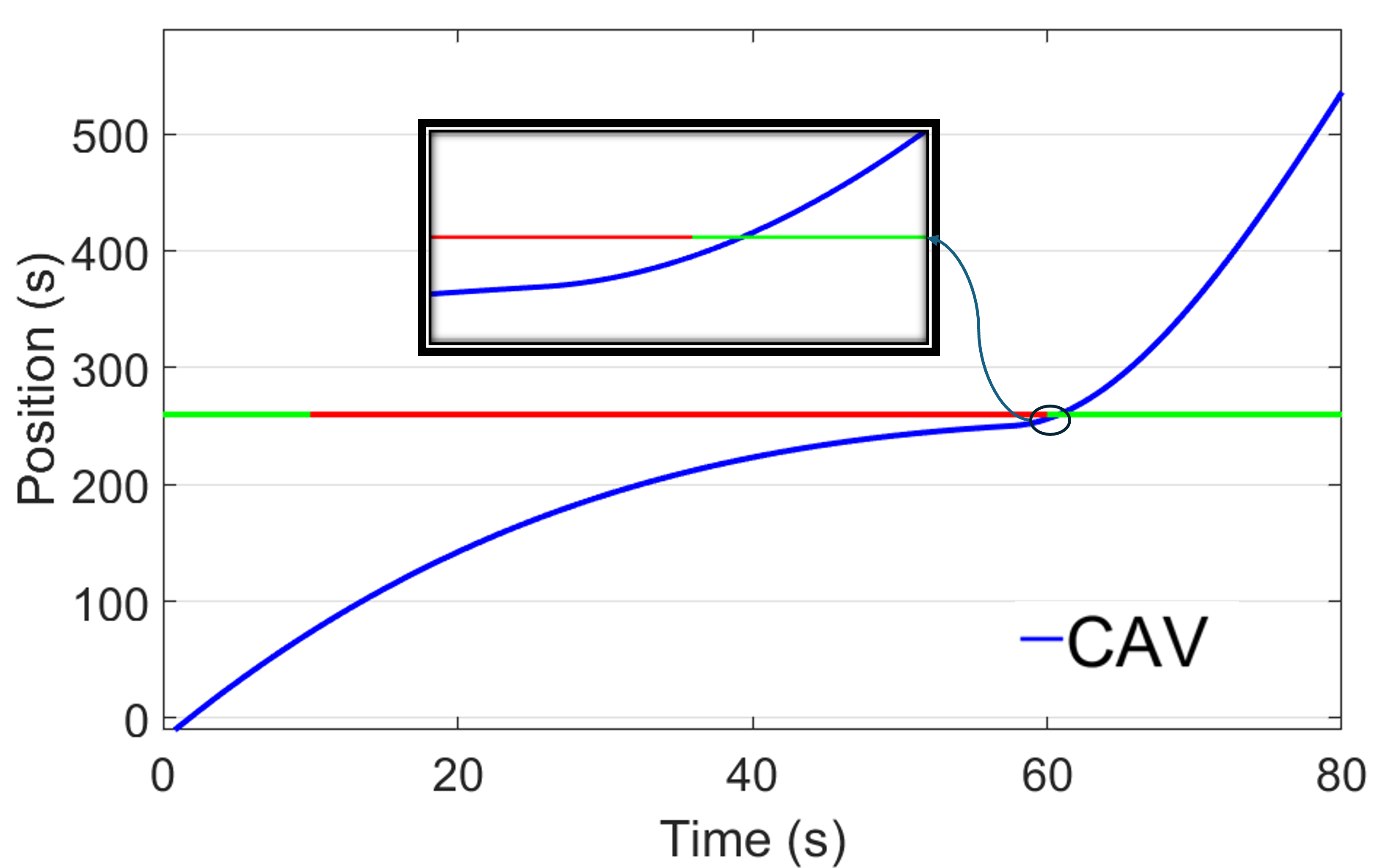}
        \caption{Standby mode and replanning instance obtained by a CAV.}
        \label{fig:result8}
    \end{subfigure}
    \hfill
    \begin{subfigure}[b]{0.49\linewidth}
        \includegraphics[width=\linewidth]{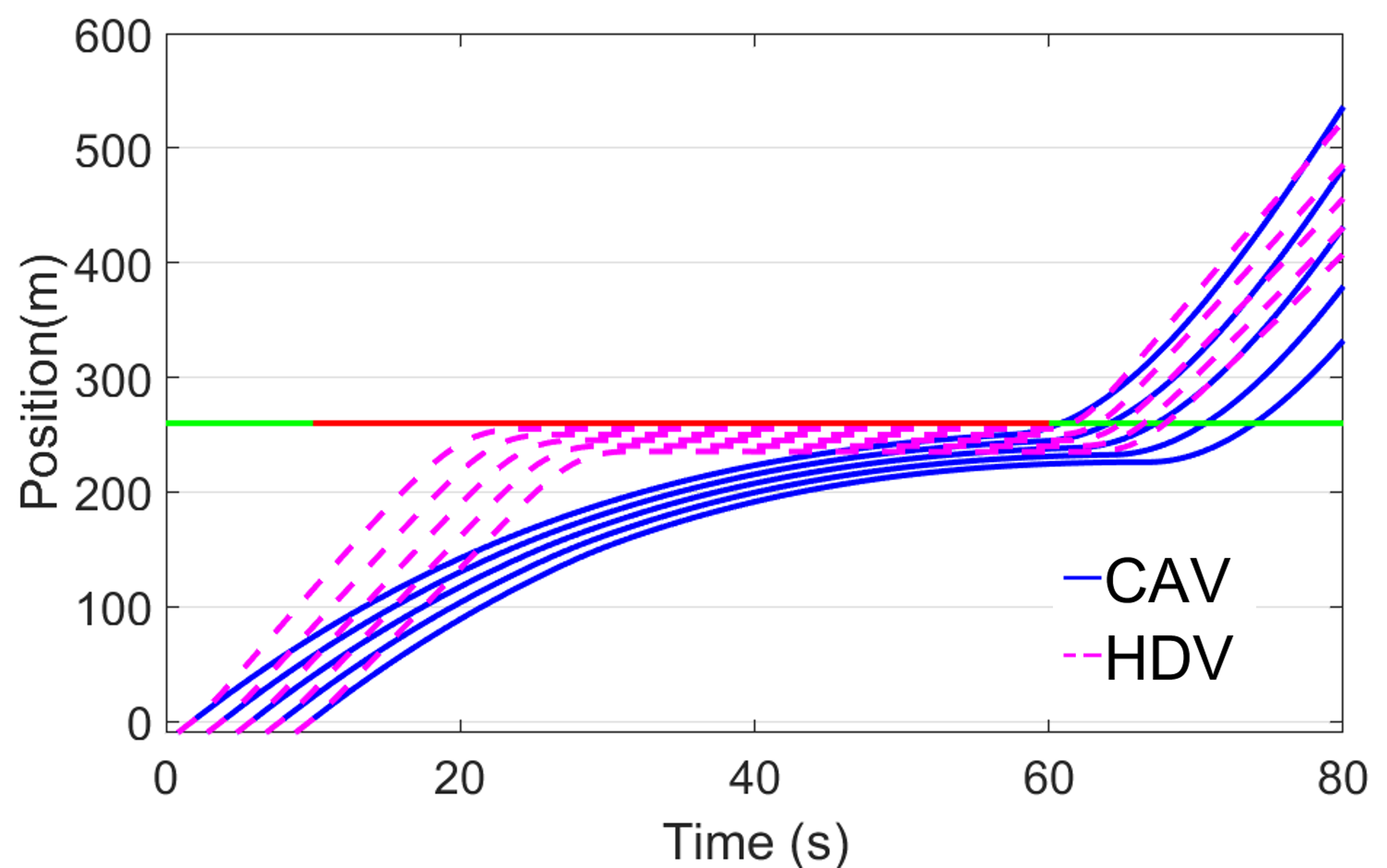}
        \caption{Difference in trajectories between CAVs and HDVs.}
        \label{fig:result9}
    \end{subfigure}
    \caption{Snapshots from different scenarios.}
    \label{fig:overall}
        \vspace{-22pt}
\end{figure}
\begin{table}[h]
\vspace{-6pt}
\centering
\begin{tabular}{|c|c|c|c|}
\hline
Signal Policy & 0$\%$ CAVs & 50$\%$ CAVs & 70$\%$ CAVs    \\
\hline
FC ($T_{\text{cycle}}=40s$) & 58.46s & 67.26s & 61.05s \\
\hline
AC ($T_{\text{cycle}}=20s$)  & 70.3s & 55.46s & 47.87s \\
\hline
AC ($T_{\text{cycle}}=30s$) & 56.77s & 59.13 s & 57.08s \\
\hline
AC ($T_{\text{cycle}}=40s$) & 70.35s & 67.36s & 67.26s \\
\hline
\end{tabular}
\caption{Average travel time per vehicle to exit the control. FT = fixed cycle, AT = adaptive cycle.}
\vspace{-15pt}
\label{tab:my_table}
\end{table}
\section{Concluding remarks} 
\vspace{-5pt}
In this study, we have developed a comprehensive control framework for signalized intersections operating within mixed traffic environments, which optimally generates trajectories for CAVs while dynamically adjusting traffic signal timings. We have introduced a mechanism that facilitates the updating of signal timings without disrupting pre-planned CAV trajectories, concurrently demonstrating that our controller can effectively influence HDVs towards adopting more energy-efficient driving behaviors. We validated the framework through simulations in MATLAB, which underscored the significant impact of traffic cycle timing across varying CAV penetration rates.  Ongoing research includes the online prediction of HDV behavior within this framework. A potential direction for future research should consider the extension of the proposed approach to accommodate scenarios involving multiple traffic phases.
\label{conlcuding remarks}

\linespread{0.99}\selectfont
\vspace{-3pt}
\bibliographystyle{ieeetr}
 \vspace{-5pt}
\bibliography{bibliography.bib,IDS_Publications_03232025,Filippos}

\end{document}